\documentclass[aps,pra,reprint,twocolumn]{revtex4-2}

\usepackage[T1]{fontenc}
\usepackage[utf8]{inputenc}
\usepackage{amsmath,amssymb,amsfonts}
\usepackage{amsthm}
\usepackage{bm}
\usepackage{graphicx}
\usepackage{mathtools}
\usepackage{bbm}
\usepackage{color}
\usepackage{enumitem}
\usepackage{soul}
\usepackage{pgfplots}
\pgfplotsset{compat=1.18}
\usepackage{pgfplots}
\usepackage{ulem}

\newcommand{\Tr}{\operatorname{Tr}}
\newcommand{\spec}{\operatorname{spec}}

\newcommand{\ket}[1]{|#1\rangle}
\newcommand{\bra}[1]{\langle #1|}
\newcommand{\proj}[1]{|#1\rangle\langle #1|}
\newcommand{\LHS}{\mathrm{LHS}}

\def\<{\langle}
\def\>{\rangle}

\newtheorem{Theorem}{Theorem}
\newtheorem{Proposition}[Theorem]{Proposition}

\newtheorem{Example}{Example}

\begin{document}

\title{Analytic construction of dimension-dependent steering witnesses from generalized equiangular measurements}

\author{Adam Rutkowski}
\email{adam.rutkowski@ug.edu.pl}
\affiliation{Institute of Theoretical Physics and Astrophysics,
Faculty of Mathematics, Physics and Informatics,
University of Gda\'{n}sk, 80-952 Gda\'{n}sk, Poland}

\author{Katarzyna Siudzi\'{n}ska}
\affiliation{Institute of Physics,
Faculty of Physics, Astronomy and Informatics,
Nicolaus Copernicus University in Toru\'{n},
ul.~Grudzi\k{a}dzka~5, 87-100 Toru\'{n}, Poland}

\begin{abstract}
In this work, we 
construct a dimension-dependent family of generalized equiangular
measurements (GEAMs) that gives explicit steering witnesses in every prime dimension $d$. The construction consists of \(d-1\) mutually unbiased
measurements and an additional normalized family of \(2d-1\) Hermitian
operators, which is a $(1,2d-1)$-POVM {(POVM denotes a positive, operator-valued measure).} For odd prime dimensions, we reduce the noncentral operators to three-dimensional invariant subspaces and solve the central operator in a two-dimensional subspace. We combine complementary Gershgorin bounds from two unitarily equivalent matrix representations to obtain a fully analytic asymmetric interval for {a parameter} that guarantees positivity. The {qubit and qutrit cases} are treated separately. Finally, using a maximally entangled realization and the
local-hidden-state bound for GEAMs, we obtain an explicit steering inequality and a sufficient criterion for quantum violation in every prime dimension.
\end{abstract}


\maketitle

\section{Introduction}
\label{sec:introduction}

Quantum steering is an asymmetric form of nonclassical correlations, lying
between entanglement and Bell nonlocality. It describes the ability of one
party to remotely prepare ensembles of quantum states for another party by
means of local measurements. Since its operational formulation
\cite{Wiseman2007}, steering has been studied both as a fundamental feature of
quantum theory and as a resource for one-sided device-independent quantum
information protocols (see Ref.~\cite{Uola2020} for a recent review). The ability to
demonstrate steering is also closely connected with measurement incompatibility
and joint measurability \cite{Heinosaari2016,Quintino2016}.

A standard way to certify steering is through a steering inequality. A
local-hidden-state (LHS) bound is obeyed by every assemblage admitting a
classical hidden-state description, while a quantum violation excludes all
such models. In contrast to the Bell inequalities, whose general structure is reviewed in
Ref.~\cite{Brunner2014}, steering inequalities can exhibit violations that
grow without bound with the local Hilbert space dimension
\cite{Cavalcanti2015}.

In Ref.~\cite{RutkowskiSiudzinska2025}, the steering functionals generated by
generalized equiangular measurements (GEAMs) were studied at an abstract
level. An analytic upper bound on their LHS value was derived solely from the
frame parameters, and high-dimensional regimes with unbounded violation were
identified. Related recent approaches use generalized local measurements in typical
steering scenarios \cite{SICMUB_preprint2}, correlation matrices associated
with equiangular tight frames \cite{Certify_EPR}, and further measurement-based
characterizations of steering \cite{Certify_EPR2}. Entropic uncertainty
relations for generalized symmetric measurements provide an alternative route to characterize steering criteria
\cite{steering_entropic}. These explicitly place the GEAM steering
functional of Ref.~\cite{RutkowskiSiudzinska2025} among the recent
developments in this direction.

The constructive problem addressed here is to realize this specific GEAM family with explicit analytic positivity bounds.
In the measurement family considered in this paper, the
nontrivial part of this problem reduces to cubic characteristic polynomials.
Although their roots can be written in closed form, those expressions are not
well suited for a transparent comparison valid in every dimension.

The purpose of this work is to provide such a realization and the associated steering witnesses. We first construct
a special family of Hermitian operators and identify the three-dimensional invariant
subspaces carrying the non-scalar spectral information of the noncentral operators, {together with a two-dimensional active subspace.} The same spectrum is
then represented in two different orthonormal bases. This redundancy is
useful: Gershgorin estimates are basis dependent, so the two representations
provide complementary bounds on the smallest and largest eigenvalues. By
combining the stronger endpoint from each basis, we obtain an explicit
asymmetric interval for the parameter that guarantees positivity of
all effects for every prime dimension $d\geq 5$. This results in a $(1,2d-1)$-POVM. {The cases with $d=2,3$ are treated separately due to a different positivity range or a different construction of mutually unbiased bases.}

Once positivity is established, the resulting $(1,2d-1)$-POVM is combined with
\(d-1\) mutually unbiased measurements to form a GEAM. Finally, we use the maximally entangled state to construct a concrete quantum
assemblage. Together with the LHS bound from
Ref.~\cite{RutkowskiSiudzinska2025}, this yields an explicit steering witness
for every prime dimension $d$. The previously established asymptotic
analysis is not repeated; the new contribution is the fully analytic
realization of the measurement family and its dimension-by-dimension steering
certification.
{For $d\geq 3$, the projective measurements already give a steering violation. The additional nontrivial POVM provides an informationally overcomplete construction; the analytic positivity bounds allow this block to be included explicitly.}

\section{Steering and local-hidden-state models}
\label{sec:steering}

Consider the standard bipartite steering scenario in which two distant
parties, Alice and Bob, share a quantum state $\rho_{AB}$ on the bipartite Hilbert space
$\mathcal{H}_A \otimes \mathcal{H}_B$.
Alice performs a local measurement labeled by $x$, with outcomes $a$,
described by a POVM
\begin{equation}
\mathcal{E}_x=\{E_{a|x}\}_a,
\qquad
E_{a|x}\geq 0,
\qquad
\sum_aE_{a|x}=\mathbb{I}_{d_A},
\label{eq:alice-povm}
\end{equation}
where $\mathbb{I}_{d_A}$ is the $d_A$-dimensional identity operator on $\mathcal{H}_A$, $d_A=\dim\mathcal{H}_A$.
Now, Bob performs full state tomography on his subsystem. 
Bob's conditional, generally unnormalized, states are
\begin{equation}
\sigma_{a|x}
=
\Tr_A\!\left[
(E_{a|x}\otimes\mathbb{I}_{d_B})\rho_{AB}
\right],
\label{eq:assemblage}
\end{equation}
where $\mathbb{I}_{d_B}$ denotes the identity operator on $\mathcal{H}_B$, $d_B=\dim\mathcal{H}_B$.
The family $\boldsymbol{\sigma}=\{\sigma_{a|x}\}_{a,x}$
is called an assemblage. By construction, it satisfies the no-signaling condition
\begin{equation}
\sum_a\sigma_{a|x}
=
\rho_B,
\qquad
\rho_B=\Tr_A\rho_{AB},
\qquad
\text{for every }x.
\label{eq:no-signalling}
\end{equation}
Equation~\eqref{eq:no-signalling} means that Bob's reduced state does not
depend on Alice's choice of measurement when Alice's outcome is ignored.

An assemblage is said to admit a local hidden state (LHS) model if it can be
written in the form
\begin{equation}
\sigma_{a|x} = \sum_\lambda p(\lambda)\,p(a|x,\lambda)\,\rho_\lambda,
\label{eq:LHS-model}
\end{equation}
where \(p(\lambda)\) is a probability distribution,
\(p(a|x,\lambda)\) are conditional probabilities, and
\(\rho_\lambda\) are states on \(\mathcal{H}_B\).
Assemblages admitting such a decomposition are called unsteerable.
If no such representation exists, the assemblage is steerable
\cite{Wiseman2007,Cavalcanti2009}.
A convenient way to detect steering is provided by steering inequalities.
Given a family of Hermitian operators $F=\{F_{a|x}\}_{a,x}$ acting on
$\mathcal{H}_B$, the corresponding steering functional is
\begin{equation}
\mathcal{S}_F(\boldsymbol{\sigma}) = \sum_{a,x} \Tr\!\left(F_{a|x}\sigma_{a|x}
\right).
\label{eq:steering-functional}
\end{equation}
Its local-hidden-state value reads
\begin{equation}
\beta_{\LHS}(F)
=
\sup_{\boldsymbol{\sigma}\in\LHS}
\mathcal{S}_F(\boldsymbol{\sigma}).
\label{eq:LHS-value}
\end{equation}
Similarly, let $\mathcal{Q}$ denote the set of quantum assemblages generated
according to Eq.~\eqref{eq:assemblage}, with the same outcome and setting
labels, and Bob's Hilbert space fixed to $\mathcal H_B$. The optimal quantum value of the steering functional is defined as
\begin{equation}
\beta_{\mathrm{Q}}(F)
=
\sup_{\boldsymbol{\sigma}\in\mathcal{Q}}
\mathcal{S}_F(\boldsymbol{\sigma}).
\label{eq:quantum-optimal-value}
\end{equation}
It suffices to exhibit a quantum assemblage
$\boldsymbol{\sigma}\in\mathcal{Q}$ with
$\mathcal{S}_F(\boldsymbol{\sigma})>\beta_{\LHS}(F)$ to certify steering.
Therefore, every unsteerable assemblage is
bounded from above,
\begin{equation}
\mathcal{S}_F(\boldsymbol{\sigma}) \leq \beta_{\LHS}(F),
\label{eq:steering-inequality}
\end{equation}
where $\beta_{\LHS}(F)$ denotes the optimal LHS bound.
Any violation of Eq.~\eqref{eq:steering-inequality} certifies steering because
the observed assemblage then lies outside the convex set of assemblages
admitting the decomposition from Eq.~\eqref{eq:LHS-model}.

In the following sections, we provide an explicit measurement
construction that realizes this regime within a restricted
setting for any prime dimension $d_A=d_B\equiv d$.

\section{Measurement scheme}
\label{sec:measurement-scheme}

In this section, we present the measurement construction that later yields explicit steering inequality violations. Using the maximal sets of mutually unbiased bases in prime dimensions $d$, we define a family of $d-1$ mutually unbiased measurements and a one-parameter $(1,2d-1)$-POVM, which possesses a highly constrained and explicitly controllable spectral structure. Together, these measurements are then used to introduce a generalized equiangular measurement (GEAM), which we later apply to witness quantum steering.

\subsection{Symmetric building blocks}
\label{subsec:symmetric-building-blocks}

Let $\mathcal{H}_A \cong \mathbb{C}^d$ with a prime $d$. Two orthonormal bases
$\mathcal{B}_\alpha=\{\ket{\psi_{\alpha,k}};\,{k=0},\ldots,{d-1}\}$
and $\mathcal{B}_\beta=\{\ket{\psi_{\beta,\ell}};\,{\ell=0},\ldots,{d-1}\}$ in $\mathcal{H}_A$
are said to be mutually unbiased if
\begin{equation}
|\langle \psi_{\alpha,k} | \psi_{\beta,\ell} \rangle|^2=\frac{1}{d},
\qquad\alpha \neq \beta.
\end{equation}
For prime dimensions $d$, there exist $d+1$ mutually unbiased bases (MUBs) $\mathcal{B}_\alpha$, which can be constructed explicitly
\cite{Ivanovic1981,Wiesniak2011}. Each MUB defines a projective measurement $\mathcal{P}_\alpha
=\{P_{\alpha,k} = \proj{\psi_{\alpha,k}};\,{k=0},\ldots,{d-1}\}$ that
satisfies the trace relations
\begin{equation}
\Tr(P_{\alpha,k}P_{\beta,\ell})
=\delta_{\alpha\beta}\delta_{k\ell}
+\frac{1}{d}(1-\delta_{\alpha\beta}).
\end{equation}

Now, we use the complete set of $d+1$ projective measurements $\mathcal{P}_\alpha$ in a prime $d\geq 3$ to define $d-1$ mutually unbiased measurements (MUMs) \cite{Kalev2014} $\mathcal{E}_\alpha=\{E_{\alpha,k};\,{k=0},\ldots,{d-1}\}$, $\alpha=1,\ldots,d-1$, which generalize the projectors onto MUB vectors to operators of any rank. In our case, the MUM operators are given by
\begin{equation}\label{eq:MUM-effects}
E_{\alpha,k}=\frac {1-t}{d} \mathbb{I}_d+tP_{\alpha,k},\qquad -\frac{1}{d-1}\leq t\leq 1,
\end{equation}
where the purity
\begin{equation}\label{kappa}
\kappa=\Tr(E_{\alpha,k}^2)=\frac{(d-1)t^2+1}{d}
\end{equation}
(see Appendix~\ref{app:MUB-MUM} for details). Next, following the methodology presented in Refs. \cite{SIC-MUB,SIC-MUB_general}, we use the remaining two sets, $\mathcal{P}_0$ and $\mathcal{P}_d$, to construct a $(1,2d-1)$-POVM $\mathcal{E}_d=\{E_j;\,j=1,\ldots,2d-1\}$, which is a POVM such that $\Tr(E_j)=w$ and $\Tr(E_j^2)=x$, and $\Tr(E_jE_\ell)=y$ for all $j\neq\ell$. {{These constants are known to satisfy
\begin{equation}
\begin{split}
w&=\frac{d}{2d-1},\qquad y=\frac{d-(2d-1)x}{2(d-1)(2d-1)},\\
z&=\frac{2}{(2d-1)^2},\qquad \frac{d}{(2d-1)^2}<x\leq \frac{d^2}{(2d-1)^2}.
\end{split}
\end{equation}}}
The $(1,2d-1)$-POVM elements read
\begin{equation}\label{Ej}
E_j=\frac{1}{2d-1}\mathbb{I}_d+\tau H_j,
\end{equation}
where
\begin{align}
H_k&=G-\sqrt{2d-1}(1+\sqrt{2d-1})G_k,\label{eq:Hk-definition}\\
H_d&=(1+\sqrt{2d-1})G,\label{eq:Hd-definition}\\
H_{d+k}&=G-\sqrt{2d-1}(1+\sqrt{2d-1})U_dG_kU_d^\dagger,\label{eq:Hdk-definition}
\end{align}
$k=1,\ldots,d-1$, are traceless Hermitian operators. In the above formulas,
\begin{equation}\label{eq:Gk-definition}
G_k=\frac{1}{\sqrt{d}(\sqrt{d}+1)}(\mathbb{I}_d+\sqrt{d}P_0)-P_k
\end{equation}
are elements of the Hermitian operator basis associated with the complete
set of MUBs. They follow from the projectors $P_k=|k\>\<k|$ onto the operational basis vectors and
\begin{equation}\label{eq:G-explicit}
G=\sum_{k=1}^{d-1}(G_k+U_dG_kU_d^\dagger)=\sqrt{d}\left(P_0+P_+-\frac 2d \mathbb{I}_d\right),
\end{equation}
where $U_d=\frac{1}{\sqrt{d}}\sum_{m,n=0}^{d-1}\omega^{mn}|m\>\<n|$ is a unitary operator, and $P_+=\frac 1d \sum_{m,n=0}^{d-1}|m\>\<n|$. Note that $U_dGU_d^\dagger=G$, and therefore, for any $k=1,\ldots,d-1$,
\begin{equation}
E_{d+k}=U_dE_kU_d^\dagger,\qquad H_{d+k}=U_dH_kU_d^\dagger.
\label{eq:H-unitary-equivalence}
\end{equation}
Moreover, $\tau$ is related to the purity $x=\Tr(E_j^2)$ via
\begin{equation}\label{eq:x-tau}
x=\frac{d}{(2d-1)^2}+2(d-1)\tau^2(\sqrt{2d-1}+1)^2,
\end{equation}
{which means that $\tau\neq 0$.}

The admissible values of the parameter $\tau$ that guarantee the semi-positivity of all $E_j$ are characterized by
\begin{equation}
-\frac{1}{2d-1}\frac{1}{\lambda_{\max}}\leq \tau\leq \frac{1}{2d-1}\frac{1}{|\lambda_{\min}|},
\end{equation}
where $\lambda_{\max}$ and $\lambda_{\min}$ are respectively the maximal and minimal values from among all eigenvalues of $H_j$, $j=1,\ldots,2d-1$. In what follows, we derive explicit bounds on these eigenvalues as functions of dimension $d$.

\subsection{Spectral structure and reduction of dimensions}
\label{subsec:spectral-reduction}

Let us analyze the structural properties of the traceless operators $H_j$, $j=1,\ldots,2d-1$. Here, we focus on $d\geq 5$. The cases of $d=2$ and $d=3$ are solved separately in the following subsection.

First, consider $H_d$, which is exactly solvable. Observe that it has its non-scalar action confined to the subspace $\mathcal{S}=\mathrm{span}\{|0\>,|s\>\}$, where
\begin{equation}
|s\>=\frac{1}{\sqrt{d-1}}\sum_{k=1}^{d-1}|k\>.
\end{equation}
In the ordered basis \(\{\ket{0},\ket{s}\}\), we find the matrix representation
\begin{equation}
H_d\big|_{\mathcal{S}}
=
\frac{1+\sqrt{2d-1}}{\sqrt d}
\begin{pmatrix}
d-1 & \sqrt{d-1}\\[2pt]
\sqrt{d-1} & d-3
\end{pmatrix},
\label{eq:Hd-2x2-block}
\end{equation}
whereas
\begin{equation}
 H_d\big|_{\mathcal{S}^\perp} = -\frac{2}{\sqrt d}(1+\sqrt{2d-1})\mathbb{I}_{d-2}.
\label{eq:Hd-complement}
\end{equation}
The two eigenvalues of $H_d\big|_{\mathcal{S}}$ read
\begin{equation}
 \lambda_\pm(H_d) = \frac{1}{\sqrt{d}}(\sqrt{d}\pm 2)(\sqrt{d}\mp 1)(1+\sqrt{2d-1}).
\label{eq:Hd-two-eigenvalues}
\end{equation}
Consequently, one has
\begin{align}
\lambda_{\min}(H_d) &= -\frac{2}{\sqrt d}(1+\sqrt{2d-1}),\label{eq:Hd-minimum}\\
\lambda_{\max}(H_d) &= \lambda_+(H_d).\label{eq:Hd-maximum}
\end{align}

Now, for the remaining $H_j$, we start with a simple observation: the form of $G$ from Eq. (\ref{eq:G-explicit}) allows us to rewrite these nontrivial operators as a linear combination of the identity operator and three rank-1 projectors. Namely, for $k=1,\ldots,d-1$, one has
\begin{equation}\label{eq:Hk-rank-one-decomposition}
H_k=\frac{1}{\sqrt{d}}\left(\mu\mathbb{I}_d+\nu P_0+\xi P_k+dP_+\right),
\end{equation}
\begin{equation}\label{eq:Hdk-rank-one-decomposition}
H_{d+k}=\frac{1}{\sqrt{d}}\left(\mu\mathbb{I}_d+\nu P_++\xi Q_k+dP_0\right),
\end{equation}
where $Q_k=U_dP_kU_d^\dagger$ and
\begin{equation}\label{eq:auxiliary-parameters}
\begin{split}
\mu&=-2-\frac{\xi}{\sqrt{d}(\sqrt{d}+1)},\\
\nu&=d-\frac{\xi}{\sqrt{d}+1},\\
\xi&=\sqrt{d(2d-1)}\big(\sqrt{2d-1}+1\big).
\end{split}
\end{equation}
Now, for each $k=1,\ldots,d-1$, we construct two subspaces,
\begin{equation}
 \mathcal{K}_k = \operatorname{span}\{\ket{0},\ket{k},\ket{w_k}\},
\label{eq:Kk-definition}
\end{equation}
\begin{equation}
 \mathcal{L}_k = \operatorname{span}\{\ket{+},\ket{v},\ket{z_k}\},
\label{eq:Lk-definition}
\end{equation}
each spanned by a set of orthonormal vectors, where $\ket{0}$ and $\ket{k}$ are elements of the operational basis and
\begin{equation}
\ket{+}
=
\frac{1}{\sqrt{d}}
\sum_{m=0}^{d-1}\ket{m},
\label{eq:plus-vector}
\end{equation}
\begin{equation}
 \ket{w_k}
 =
 \frac{1}{\sqrt{d-2}}
 \sum_{\substack{m=1\\m\neq k}}^{d-1}\ket{m},
\label{eq:wk-definition}
\end{equation}
\begin{equation}
 \ket{v}
 =
 \frac{\sqrt{d}\ket{0}-\ket{+}}{\sqrt{d-1}},
\label{eq:v-definition}
\end{equation}
\begin{equation}
 \ket{z_k}
 =
 \sqrt{\frac{d-1}{d-2}}
 \left(
 U_d\ket{k}-\frac{1}{\sqrt{d-1}}\ket{v}
 \right),
\label{eq:zk-definition}
\end{equation}
This allows us to decompose $H_k$ and $H_{d+k}$ into a direct sum.

\begin{Proposition}[Invariant-block decomposition]
\label{prop:spectral-reduction}
For every \(k=1,\ldots,d-1\), the subspaces \(\mathcal{K}_k\) and
\(\mathcal{L}_k\) are invariant under \(H_k\) and \(H_{d+k}\),
respectively. Moreover,
\begin{equation}
 H_k \simeq B^{(1)}_k\oplus \frac{\mu}{\sqrt{d}}\mathbb{I}_{d-3},
\label{eq:Hk-direct-sum}
\end{equation}
\begin{equation}
 H_{d+k} \simeq B^{(2)}_k\oplus \frac{\mu}{\sqrt{d}}\mathbb{I}_{d-3},
\label{eq:Hdk-direct-sum}
\end{equation}
with $B^{(1)}_k$ and $B^{(2)}_k$ being $3\times 3$ Hermitian matrices,
\begin{equation}
 B^{(1)}_k =\frac{1}{\sqrt{d}}
\begin{pmatrix}
\mu+\nu+1 & 1 & \eta\\
1 & \mu+\xi+1 & \eta\\
\eta & \eta & \mu+\eta^2
\end{pmatrix},
\label{eq:first-3x3-block}
\end{equation}
\begin{equation}
 B^{(2)}_k = \frac{1}{\sqrt{d}}
 \begin{pmatrix}
 \mu+\nu+1 & \chi & 0 \\
\chi & \mu+\zeta+\chi^2 & \eta\zeta \\
 0 & \eta\zeta & \mu+\eta^2\zeta
 \end{pmatrix}
\label{eq:second-3x3-block}
\end{equation}
in the ordered bases
\(\{\ket{0},\ket{k},\ket{w_k}\}\) and
\(\{\ket{+},\ket{v},\ket{z_k}\}\), respectively, where
\begin{equation}\label{etas}
\eta=\sqrt{d-2},\quad\chi=\sqrt{d-1},\quad\zeta=\xi/\chi^2,
\end{equation}
and the remaining parameters are given in Eq. (\ref{eq:auxiliary-parameters}).
\end{Proposition}

\begin{proof}
{Observe that vectors \(\ket{0}\), \(\ket{k}\), and \(\ket{w_k}\) are orthonormal, and \(\ket{+}\) can be decomposed into
\begin{equation}
 \ket{+}
 =
 \frac{1}{\sqrt{d}}
 \left(
 \ket{0}+\ket{k}+\sqrt{d-2}\ket{w_k}
 \right).
\label{eq:plus-decomposition-first-basis}
\end{equation}
Hence, the ranges of \(P_0\), \(P_k\), and \(P_+\) lie in
\(\mathcal{K}_k\). Now, Eq.~\eqref{eq:Hk-rank-one-decomposition}
implies
\begin{equation}
 H_k\big|_{\mathcal{K}_k^\perp}
 =
 \frac{\mu}{\sqrt d}\mathbb{I}_{d-3}.
\label{eq:Hk-complement-action}
\end{equation}
Substituting \(P_+=\ket{+}\bra{+}\) with $\ket{+}$ from Eq.~\eqref{eq:plus-decomposition-first-basis} into Eq.~\eqref{eq:Hk-rank-one-decomposition} produces $B^{(1)}_k$ from
Eq.~\eqref{eq:first-3x3-block}.}

{
Next, from Eqs.~\eqref{eq:v-definition} and \eqref{eq:zk-definition}, it follows  that \(\ket{+}\), \(\ket{v}\), and \(\ket{z_k}\) are orthonormal. In this
basis,
\begin{equation}
 \ket{0}
 =
 \frac{1}{\sqrt{d}}\ket{+}
 +
 \sqrt{\frac{d-1}{d}}\ket{v},
\label{eq:zero-decomposition-second-basis}
\end{equation}
\begin{equation}
 U_d\ket{k}
 =
 \frac{1}{\sqrt{d-1}}\ket{v}
 +
 \sqrt{\frac{d-2}{d-1}}\ket{z_k}.
\label{eq:phi-decomposition-second-basis}
\end{equation}
Thus, the ranges of \(P_+\), \(P_0\), and \(Q_k=U_dP_kU_d^\dagger\) lie in
\(\mathcal{L}_k\), and
\begin{equation}
 H_{d+k}\big|_{\mathcal{L}_k^\perp}
 =
 \frac{\mu}{\sqrt d}\mathbb{I}_{d-3}.
\label{eq:Hdk-complement-action}
\end{equation}
Finally, substituting Eqs.~\eqref{eq:zero-decomposition-second-basis} and
\eqref{eq:phi-decomposition-second-basis} into Eq.~\eqref{eq:Hdk-rank-one-decomposition} gives
Eq.~\eqref{eq:second-3x3-block}.}
\end{proof}

{Finally, Eq.~\eqref{eq:H-unitary-equivalence} implies that \(H_k\) and
\(H_{d+k}\) are unitarily similar. Since both contain the same scalar block
\(\frac{\mu}{\sqrt d}\mathbb{I}_{d-3}\), their blocks are isospectral,}
\begin{equation}
 \spec\bigl(B^{(1)}_k\bigr)
 =
 \spec\bigl(B^{(2)}_k\bigr).
\label{eq:block-isospectrality}
\end{equation}
Therefore, there exists a real symmetric transformation matrix
\begin{equation}
\Omega=\frac{1}{\sqrt{d-1}}\begin{pmatrix}
\sqrt{d-1} & 0 & 0 \\
0 & 1 & \sqrt{d-2} \\
0 & \sqrt{d-2} & -1
\end{pmatrix}
\end{equation}
such that $B^{(2)}_k=\Omega B^{(1)}_k \Omega$.
This does not render either matrix
representation redundant. Despite $B^{(1)}_k$ and $B^{(2)}_k$ being real symmetric $3\times 3$ matrices, it proves remarkably difficult to provide simplified formulas for their eigenvalues. Therefore, we instead analyze their lower and upper bounds, which are not invariant under the change of basis.

\subsection{Analytical positivity interval}
\label{subsec:analytic-positivity}

To obtain the eigenvalue bounds for $B^{(1)}_k$ and $B^{(2)}_k$ from Proposition \ref{prop:spectral-reduction}, we apply the following theorem.

\begin{Proposition}[Gershgorin theorem \cite{Gershgorin1931,Varga,HornJohnson2013}]
\label{prop:Gershgorin}
Let \(A=(a_{ij})\in M_d(\mathbb{C})\) and define a radius
\begin{equation}
 r_i=\sum_{j\neq i}|a_{ij}|,\qquad i=1,\ldots,d.
\label{eq:main-Gershgorin-radius}
\end{equation}
Every eigenvalue of \(A\) belongs to the Gershgorin disc
\begin{equation}
 \bigcup_{i=1}^{d}
 \{z\in\mathbb{C}:|z-a_{ii}|\leq r_i\}.
\label{eq:main-Gershgorin-discs}
\end{equation}
If \(A\) is Hermitian, then its minimal and maximal eigenvalues are bounded by
\begin{equation}
\begin{split}
 \lambda_{\min}(A)&\geq\min_i(a_{ii}-r_i),\\
 \lambda_{\max}(A)&\leq\max_i(a_{ii}+r_i).
\label{eq:main-Gershgorin-bounds}
\end{split}
\end{equation}
\end{Proposition}

Notably, the Gershgorin discs are not invariant under a change
of basis. Consequently, applying the Gershgorin theorem to two matrices with the same spectrum can in general produce two different admissible ranges for their eigenvalues. Therefore, we perform the spectral analysis for both $B^{(1)}_k$ and $B^{(2)}_k$ to obtain uniform lower and upper bounds on the eigenvalues of all $H_j$. Detailed derivations are provided in Appendix \ref{app:Gershgorin}, with the most important steps presented below.

Starting from Eq. (\ref{eq:main-Gershgorin-bounds}), we find that
\begin{align}
 \lambda_{\min}\bigl(B^{(1)}_k\bigr)&\geq m_d^{(1)},
 &
 \lambda_{\max}\bigl(B^{(1)}_k\bigr)&\leq M_d^{(1)},
\label{eq:first-basis-Gershgorin}\\
 \lambda_{\min}\bigl(B^{(2)}_k\bigr)&\geq m_d^{(2)},
 &
 \lambda_{\max}\bigl(B^{(2)}_k\bigr)&\leq M_d^{(2)},
\label{eq:second-basis-Gershgorin}
\end{align}
where, for any prime $d\geq 5$,
\begin{align}
m_d^{(1)}&=-\frac{d+1+\sqrt{d-2}+\sqrt{2d-1}}{\sqrt{d}},\label{eq:md1-definition}\\
M_d^{(1)}&=\frac{\mu+\xi+\sqrt{d-2}+2}{\sqrt{d}},\label{eq:Md1-definition}\\
m_d^{(2)}&=-\frac{d+\sqrt{d-1}+\sqrt{2d-1}}{\sqrt{d}},\label{eq:md2-definition}\\
M_d^{(2)}&=\frac{\mu+\sqrt{d-2}(\sqrt{d-2}+1)\zeta}{\sqrt{d}},\label{eq:Md2-definition}
\end{align}
with the parameters $\mu$, $\zeta$ given in Eqs. (\ref{eq:auxiliary-parameters}) and (\ref{etas}), respectively.
Because $B^{(1)}_k$ and $B^{(2)}_k$ are isospectral, both estimates in
Eqs.~\eqref{eq:first-basis-Gershgorin} and
\eqref{eq:second-basis-Gershgorin} hold simultaneously for the same
spectrum. Hence,
\begin{equation}
\begin{split}
 \lambda_{\min}(B^{(1)}_k)
 &=
 \lambda_{\min}(B^{(2)}_k)\\&
 \geq
 \max\{m_d^{(1)},m_d^{(2)}\}
 =m_d^{(2)},
\end{split}
\label{eq:combined-lower-Gershgorin}
\end{equation}
\begin{equation}
\begin{split}
 \lambda_{\max}(B^{(1)}_k)
 &=
 \lambda_{\max}(B^{(2)}_k)\\&
 \leq
 \min\{M_d^{(1)},M_d^{(2)}\}
 =M_d^{(1)}.
\end{split}
\label{eq:combined-upper-Gershgorin}
\end{equation}
In addition, one has
\begin{equation}
m_d^{(2)}\leq \frac{\mu}{\sqrt{d}}\leq M_d^{(1)},
\end{equation}
and therefore the minimal and maximal eigenvalues of $H_j$, $j\neq d$, are bounded by
\begin{equation}\label{Hjj}
\begin{split}
\lambda_{\min}(H_j)\geq m_d^{(2)},\qquad
\lambda_{\max}(H_j)\leq M_d^{(1)}.
\end{split}
\end{equation}
The final analytical estimations follow from comparing the above results to Eqs. (\ref{eq:Hd-minimum}--\ref{eq:Hd-maximum}),
\begin{equation}
\begin{split}
\lambda_{\min}&\geq\min\{m_d^{(2)},\lambda_{\min}(H_d)\}=m_d^{(2)},\\
\lambda_{\max}&\leq\max\{M_d^{(1)},\lambda_{\max}(H_d)\}=M_d^{(1)},
\end{split}
\end{equation}
which means that $m_d^{(2)}$ and $M_d^{(1)}$ provide the eigenvalue bounds for all $H_j$, $j=1,\ldots,2d-1$.

\begin{Theorem}[Positivity bounds for $d\geq 5$]
\label{thm:tau-bound}
For any prime $d\geq 5$, if
\begin{equation}
 -\frac{1}{(2d-1)M_d^{(1)}}
 \leq
 \tau
 \leq
 -\frac{1}{(2d-1)m_d^{(2)}},
\label{eq:tau-interval}
\end{equation}
then all $E_j$ in Eq. (\ref{Ej}) are positive semi-definite, and hence they form a valid $(1,2d-1)$-POVM.
\end{Theorem}

\subsection{Low-dimensional cases}
\label{subsec:low-dimensional}

In Theorem \ref{thm:tau-bound}, we have provided the range of $\tau$ that guarantees semi-positivity of all $E_j$ for any prime $d\geq 5$. Here, we consider the remaining two cases, $d=2$ and $d=3$.

First, let us take $d=3$. The entire analysis of $(1,2d-1)$-POVMs from the previous subsections holds up until Eq. (\ref{eq:second-basis-Gershgorin}).
The differences start with the value of $M_d^{(2)}$, which reads instead
\begin{equation}
M_d^{(2)}=\frac{d-1+\sqrt{d-1}+\mu+(\sqrt{d-2}+1)\zeta}{\sqrt{d}}.
\end{equation}
Next, we have
\begin{equation}
\begin{split}
 \lambda_{\min}(B^{(1)}_k)
 &=
 \lambda_{\min}(B^{(2)}_k)\\&
 \geq
 \max\{m_d^{(1)},m_d^{(2)}\}
 =m_d^{(2)},
\end{split}
\end{equation}
\begin{equation}
\begin{split}
 \lambda_{\max}(B^{(1)}_k)
 &=
 \lambda_{\max}(B^{(2)}_k)\\&
 \leq
 \min\{M_d^{(1)},M_d^{(2)}\}
 =M_d^{(1)},
\end{split}
\end{equation}
which means that the eigenvalue estimates are {again}
\begin{equation}
\begin{split}
\lambda_{\min}&\geq\min\{m_d^{(2)},\lambda_{\min}(H_d)\}=m_d^{(2)},\\
\lambda_{\max}&\leq\max\{M_d^{(1)},\lambda_{\max}(H_d)\}=M_d^{(1)}.
\end{split}
\end{equation}

\begin{Theorem}[Positivity bounds for $d=3$]
\label{thm:tau-bound3}
For $d=3$, if
\begin{equation}
 -\frac{1}{(2d-1)M_d^{(1)}}
 \leq \tau \leq
 -\frac{1}{(2d-1)m_d^{(2)}},
\label{eq:tau-interval3}
\end{equation}
then all $E_j$ in Eq. (\ref{Ej}) are positive semi-definite, and hence they form a valid $(1,2d-1)$-POVM.
\end{Theorem}

Even though the above results are analytically valid, the eigenvalue estimate through Gershgorin discs fails to detect one structural property that differs the case $d=3$ from $d\geq 5$. Numerical calculations show that $\lambda_{\min}=\lambda_{\min}(H_d)$. This behavior is presented in the following example.

\begin{Example}
For $d=3$, the matrices from Eq. (\ref{eq:first-3x3-block}) have the form ($k=1,2$)
\begin{equation}
B_k^{(1)}
=\frac{1}{\sqrt{3}}
\begin{pmatrix}
\theta_1 & 1 & 1 \\
1 & \theta_2 & 1 \\
1 & 1 & \theta_3
\end{pmatrix},
\end{equation}
where
\begin{align}
\theta_1&=-(3+\sqrt{5}),\\
\theta_2&=\frac{9+4\sqrt{3}+2\sqrt{5}+\sqrt{15}}{1+\sqrt{3}},\\
\theta_3&=-\frac{6+\sqrt{3}+\sqrt{5}}{1+\sqrt{3}}.
\end{align}
Straightforward calculations allow us to numerically obtain its eigenvalues,
\begin{align*}
\lambda_1(B_k^{(1)})&\simeq-3.31088,\\
\lambda_2(B_k^{(1)})&\simeq-1.91146,\\
\lambda_3(B_k^{(1)})&\simeq 5.22235.
\end{align*}
Finally, the eigenvalues of $H_d$ from Eq. (\ref{eq:Hd-2x2-block}) reduce to
\begin{align*}
\lambda_1(H_d)&=\frac{1}{\sqrt{3}}(\sqrt{3}+1)(\sqrt{3}-2)(1+\sqrt{5})\simeq -1.36772,\\
\lambda_2(H_d)&=\frac{1}{\sqrt{3}}(\sqrt{3}-1)(\sqrt{3}+2)(1+\sqrt{5})\simeq 5.10441,\\
\lambda_3(H_d)&=-\frac{2}{\sqrt{3}}(1+\sqrt{5})\simeq -3.73669.
\end{align*}
Using the approximated values, we easily find the extremal eigenvalues
\begin{equation}
\lambda_{\max}=\lambda_3(B_k^{(1)}),\qquad \lambda_{\min}=\lambda_3(H_d)
\end{equation}
{and the approximate admissible range of the parameter $\tau$,
\begin{equation}
-0.038297<\tau<0.053523.
\end{equation}
By comparison, the estimates from Theorem \ref{thm:tau-bound3} produce the approximate interval with endpoints
\begin{equation}
-0.031826<\tau<0.052090,
\end{equation}}
where
\begin{align*}
m_d^{(2)}&=-\frac{3+\sqrt{2}+\sqrt{5}}{\sqrt{3}}\simeq-3.83954,\\
M_d^{(1)}&=\frac{11+6\sqrt{3}+2\sqrt{5}+\sqrt{15}}{\sqrt{3}(\sqrt{3}+1)}\simeq 6.28426.
\end{align*}
Indeed, $m_d^{(2)}\leq\lambda_3(H_d)$ and the actual range of admissible $\tau$ is slightly larger than its Gershgorin estimate.
\end{Example}

In the qubit case of $d=2$, the mutually unbiased bases are constructed differently than for prime $d\geq 3$ (see Appendix A). In our approach, we now have a single MUM $\mathcal{E}_1=\{E_{1,0},E_{1,1}\}$ with
\begin{equation}
E_{1,k}=\frac{1-t}{2}\mathbb{I}_2+t|k\>\<k|,\qquad -1\leq t\leq 1,
\end{equation}
and a $(1,3)$-POVM $\mathcal{E}_2=\{E_1,E_2,E_3\}$.
Using the general methods presented in Eqs. (\ref{Ej}--\ref{eq:Gk-definition}), we construct the elements
\begin{equation}\label{Ej2}
\begin{split}
E_{1}&=\frac 13 \mathbb{I}_2+\frac{\tau}{\sqrt{2}}
\Big[\sigma_2-(2+\sqrt{3})\sigma_1\Big],\\
E_{2}&=\frac 13 \mathbb{I}_2+\frac{\tau}{\sqrt{2}}
\Big[\sigma_1-(2+\sqrt{3})\sigma_2\Big],\\
E_{3}&=\frac 13 \mathbb{I}_2+\frac{\tau}{\sqrt{2}}(1+\sqrt{3})
(\sigma_1+\sigma_2).
\end{split}
\end{equation}
The semi-positivity conditions for $E_j$, $j=1,2,3$, are controlled by the following theorem.

\begin{Theorem}[Positivity bounds for $d=2$]
\label{thm:tau-bound2}
For $d=2$, if
\begin{equation}
-\frac{1}{3(1+\sqrt{3})} \leq \tau \leq \frac{1}{3(1+\sqrt{3})},
\label{eq:tau-interval2}
\end{equation}
then all $E_j$ in Eq. (\ref{Ej2}) are positive semi-definite, and hence they form a valid $(1,3)$-POVM.
\end{Theorem}

This can be easily checked as all the elements of $\mathcal{E}_2$ have the same spectrum,
\begin{equation}
\lambda_\pm=\frac 13 \pm (1+\sqrt{3})\tau.
\label{eq:qubit-effect-spectrum}
\end{equation}
{For $t=1$, the MUM is projective. At $\tau=1/[3(1+\sqrt3)]$, the additional POVM has rank-one effects with eigenvalues $0$ and $2/3$ \cite{SIC-MUB_general}.}

\subsection{Generalized equiangular measurement}
\label{subsec:GEAM-subclasses}

In this section, we obtain an informationally overcomplete POVM that is less symmetric than a maximal set of $d+1$ MUMs but nevertheless can be constructed in any prime dimension $d$. Following ref. \cite{GEAM}, we use a set of $d-1$ MUMs $\mathcal{E}_\alpha=\{E_{\alpha,k}\}$, $\alpha=1,\ldots,d-1$, and a $(1,2d-1)$-POVM $\mathcal{E}_d=\{E_j\}$ to introduce a single POVM $\mathcal{F}=\{F_{\alpha,k}\}$ via
\begin{equation}
\begin{split}
F_{\alpha,k}&=\gamma_\alpha E_{\alpha,k},\quad k=0,\ldots,d-1,\\
F_{d,j}&=\gamma_d E_j,\quad j=1,\ldots,2d-1,
\end{split}
\end{equation}
where $\gamma_\alpha>0$ for all $\alpha=1,\ldots,d$, and $\sum_{\alpha=1}^d\gamma_\alpha=1$. Observe that $F_{\alpha,k}$ are semi-positive by construction. Now, $\mathcal{F}$ belongs to a class of POVMs called generalized equiangular measurements (GEAMs). Its equiangularity is manifested in a relation between elements with a fixed index $\alpha$,
\begin{equation}
\Tr(F_{\alpha,k}F_{\alpha,\ell})=c_\alpha \Tr(F_{\alpha,k})\Tr(F_{\alpha,\ell}),\qquad k\neq\ell,
\end{equation}
with a nonnegative constant $c_\alpha$.
For any $\alpha\neq\beta$, one also recovers a complementary relation (mutual unbiasedness),
\begin{equation}
\Tr(F_{\alpha,k}F_{\beta,\ell})=\frac 1d \Tr(F_{\alpha,k})\Tr(F_{\beta,\ell}).
\end{equation}
In general, the operators $F_{\alpha,k}$ do not have a uniform trace. However, taking
\begin{equation}
\begin{split}
\gamma_\alpha&=\frac{d}{d^2+d-1},\quad \alpha=1,\ldots,d-1,\\
\gamma_d&=\frac{2d-1}{d^2+d-1},
\end{split}
\end{equation}
results in $\Tr(F_{\alpha,k})=d/(d^2+d-1)$ for all $F_{\alpha,k}\in\mathcal{F}$. Then, $F_{\alpha,k}$ form a generalized equiangular measurement (GEAM) characterized by the trace relations
\begin{equation}\label{coeff}
\begin{split}
\Tr(F_{\alpha,k})&=a_\alpha\equiv\frac{d}{d^2+d-1},\\
\Tr(F_{\alpha,k}^2)&=a_\alpha^2 b_\alpha,\\
\Tr(F_{\alpha,k}F_{\alpha,\ell})&=a_\alpha^2 c_\alpha,\qquad k\neq\ell,\\
\Tr(F_{\alpha,k}F_{\beta,\ell})&=\frac 1d a_\alpha a_\beta,\qquad \alpha\neq\beta,
\end{split}
\end{equation}
where
\begin{equation}
b_\alpha=\kappa,\qquad c_\alpha=\frac{1-\kappa}{d-1},\qquad \alpha=1,\ldots,d-1,
\end{equation}
\begin{equation}
b_d=\frac{(2d-1)^2x}{d^2},\quad c_d=\frac{2d-1}{2d^2(d-1)}[d-(2d-1)x].
\end{equation}

\section{Dimension-dependent steering witnesses}
\label{sec:steering-witnesses}

In general, any GEAM forms a set $\mathcal{F}=\{F_{\alpha,k}:\,\alpha=1,\ldots,N;\,k=1,\ldots,M_\alpha\}$, where $N$ is the number of POVMs it has been constructed from and $M_\alpha$ are the numbers of operators inside those POVMs. 
Recall that it is normalized as follows,
\begin{equation}
\sum_{\alpha=1}^N\sum_{k=1}^{M_\alpha}F_{\alpha,k}
=\left(\sum_{\alpha=1}^N\gamma_\alpha\right)\mathbb{I}_d
=\mathbb{I}_d.
\end{equation}
{Moreover, for a general GEAM, $a_\alpha=d\gamma_\alpha/M_\alpha$ and $\Tr(F_{\alpha,k}^2)=a_\alpha^2b_\alpha$ are independent of $k$. In this section, outcomes are relabeled as $k=1,\ldots,M_\alpha$.}
Now, $\mathcal{F}$ defines a steering functional
\begin{equation}\label{eq:lhs-eval}
\mathcal{S}_{\mathcal{F}} = \sum_{\alpha=1}^N\sum_{k=1}^{M_\alpha} \mathrm{Tr}\left(F_{\alpha,k} \sigma_{\alpha,k}\right).
\end{equation}
For an LHS assemblage,
\begin{equation}
\sigma_{\alpha,k} = \sum_{\lambda} p(\lambda)p(k|\alpha,\lambda) \rho_\lambda
\end{equation}
for a probability distribution $p(\lambda)$ and a set of quantum states $\{\rho_\lambda\}$.

\begin{Proposition}[LHS upper bound for a GEAM functional~\cite{RutkowskiSiudzinska2025}]
\label{prop:GEAM-LHS-bound}
For the steering functional generated by $\mathcal{F}$, every LHS
assemblage satisfies
\begin{equation}
 \mathcal{S}_{\mathcal{F}}
 \leq
 \xi_+(\mathcal{F}),
\label{eq:generic-steering-bound}
\end{equation}
where
\begin{equation}
 \xi_+(\mathcal{F}) = \frac{1}{d} \left[ \sum_{\alpha=1}^{N}a_\alpha
 + \sqrt{ (d-1) \sum_{\alpha=1}^{N}a_\alpha^2 \left(db_\alpha-1\right)} \right].
\label{eq:generic-xi-plus}
\end{equation}
\end{Proposition}

Note that the quantity \(\xi_+(\mathcal{F})\) is an analytic upper bound on the optimal LHS value.

We now specify a quantum realization. Let us take the maximally entangled Bell state
\begin{equation}
 \ket{\Phi_d}
 =
 \frac{1}{\sqrt{d}}
 \sum_{n=0}^{d-1}\ket{n}\otimes\ket{n}
\label{eq:maximally-entangled-state}
\end{equation}
and the corresponding projector \(\Phi_d=\proj{\Phi_d}\). For each block $\alpha=1,\ldots,N$, Alice performs the POVM
\begin{equation}
 A_{\alpha,k}
 =
 \frac{F_{\alpha,k}^{\mathsf{T}}}{\gamma_\alpha},
\label{eq:Alice-GEAM-POVM}
\end{equation}
where the transpose is taken in the computational basis. 
The corresponding assemblage on Bob's side
\begin{equation}
 \sigma_{\alpha,k}^{\Phi}
 =
 \Tr_A\!\left[
 (A_{\alpha,k}\otimes\mathbb{I}_d)\Phi_d
 \right]
 =
 \frac{F_{\alpha,k}}{d\gamma_\alpha}.
\label{eq:maximally-entangled-assemblage}
\end{equation}
is a standard no-signalling assemblage, since for every block
\begin{equation}
 \sum_{k=1}^{M_\alpha}\sigma_{\alpha,k}^{\Phi}
 =
 \frac{1}{d\gamma_\alpha}
 \sum_{k=1}^{M_\alpha}F_{\alpha,k}
 =
 \frac{1}{d}\mathbb{I}_d.
\label{eq:maximally-entangled-no-signalling}
\end{equation}
Thus, Bob's reduced state is \(\rho_B=\mathbb{I}_d/d\), independently of
Alice's measurement setting.
For the functional generated by the same GEAM, the value attained by this maximally entangled realization is given by
\begin{align}
 Q_{\Phi}(\mathcal{F}) := \mathcal{S}_{\mathcal{F}}(\boldsymbol{\sigma}^{\Phi}) &=
 \sum_{\alpha=1}^{N}\sum_{k=1}^{M_\alpha}
 \Tr\!\left( F_{\alpha,k}\sigma_{\alpha,k}^{\Phi} \right)
\nonumber\\
 &= \sum_{\alpha=1}^{N} a_\alpha b_\alpha.
\label{eq:generic-quantum-value}
\end{align}
This attained value provides a lower bound on the optimum defined in
Eq.~\eqref{eq:quantum-optimal-value}:
\begin{equation}
Q_{\Phi}(\mathcal{F})\leq\beta_{\mathrm{Q}}(\mathcal{F}).
\label{eq:attained-quantum-bound}
\end{equation}

Recall that, for the GEAM \(\mathcal{F}\) introduced in Section~\ref{subsec:GEAM-subclasses},
\begin{equation}
 M_\alpha=d,
 \quad \alpha=1,\ldots,d-1,
 \qquad
 M_d=2d-1.
\label{eq:our-block-sizes}
\end{equation}
For the construction with $d\geq3$, define
\begin{equation}
 Z(d,t,\tau)
 =
 (d-1)t^2
 +
\frac{2(2d-1)^2(1+\sqrt{2d-1})^2}{d}\tau^2.
\label{eq:Z-definition}
\end{equation}
Substituting this into Eq.~\eqref{eq:generic-xi-plus} gives
\begin{equation}
 \xi_+(d,t,\tau)
 =
 \frac{d+(d-1)\sqrt{Z(d,t,\tau)}}{d^2+d-1}.
\label{eq:explicit-xi-plus}
\end{equation}
{Now, Eq.~\eqref{eq:generic-quantum-value} simplifies to}
\begin{equation}
 Q_{\Phi}(d,t,\tau)
 =
 \frac{d+(d-1)Z(d,t,\tau)}{d^2+d-1}.
\label{eq:explicit-quantum-value}
\end{equation}

\begin{Theorem}[Dimension-dependent steering witness]
\label{thm:dimension-witness}
Let \(d\geq5\) be prime, let \(t\) satisfy
Eq.~\eqref{eq:MUM-effects}, and let \(\tau\) satisfy
Eq.~\eqref{eq:tau-interval}. If
\begin{equation}
 Z(d,t,\tau)>1,
\label{eq:steering-criterion-Z}
\end{equation}
then
\begin{equation}
 \beta_{\LHS}(\mathcal{F})
 \leq \xi_+(d,t,\tau)
 < Q_{\Phi}(d,t,\tau)
 \leq \beta_{\mathrm{Q}}(\mathcal{F}),
\label{eq:LHS-bound-chain}
\end{equation}
and the assemblage in Eq.~\eqref{eq:maximally-entangled-assemblage}
violates the GEAM steering inequality.
\end{Theorem}

\begin{proof}
{
Starting from Eqs.~\eqref{eq:explicit-xi-plus} and
\eqref{eq:explicit-quantum-value}, one finds
\begin{equation}
\begin{split}
 Q_{\Phi}&(d,t,\tau)-\xi_+(d,t,\tau)
 \\&=\frac{d-1}{d^2+d-1}
 \left[Z(d,t,\tau)-\sqrt{Z(d,t,\tau)}\right],
\end{split}
\label{eq:quantum-minus-LHS-bound}
\end{equation}
which is positive precisely when \(Z(d,t,\tau)>1\).
The remaining inequalities follow directly from
Proposition~\ref{prop:GEAM-LHS-bound} and
Eq.~\eqref{eq:attained-quantum-bound}.
}
\end{proof}

With the normalization
\begin{equation}
 \widehat{F}_{\alpha,k}
 =\frac{F_{\alpha,k}}{Q_{\Phi}(d,t,\tau)},
\label{eq:normalized-functional}
\end{equation}
the maximally entangled realization attains
\begin{equation}
 \widehat{\mathcal{S}}_{\Phi}=1,
\label{eq:normalized-quantum-value}
\end{equation}
whereas, for \(Z(d,t,\tau)>1\), every LHS assemblage satisfies
\begin{equation}
 \widehat{\mathcal{S}}_{\LHS}
 \leq\frac{\xi_+(d,t,\tau)}{Q_{\Phi}(d,t,\tau)}<1.
\label{eq:normalized-LHS-bound}
\end{equation}
An explicit admissible choice is
\begin{equation}
 t=1,\qquad
 \tau=-\frac{1}{(2d-1)m_d^{(2)}},
\label{eq:explicit-witness-choice}
\end{equation}
{where $\tau$ is the upper endpoint from Eq.~\eqref{eq:tau-interval}}. Since
\begin{equation}
 Z(d,1,\tau)\geq d-1>1
\label{eq:Z-explicit-choice}
\end{equation}
for every prime \(d\geq5\) and every admissible \(\tau\), this gives
a steering witness in each such dimension.

For $d=3$, take $t=1$ and any nonzero $\tau$ in
Eq.~\eqref{eq:tau-interval3}. The same expressions for $\xi_+$ and $Q_\Phi$
apply, and $Z(3,1,\tau)\geq2>1$ proves the violation.
{
Finally, for $d=2$, take the maximal values of
$t=1$ and $\tau=1/[3(1+\sqrt3)]$. Fix the weights
$\gamma_1=2/5$, $\gamma_2=3/5$, so that the GEAM parameters are
$a_1=a_2=2/5$, $b_1=b_2=1$. Then, Eqs.~\eqref{eq:generic-xi-plus}
and \eqref{eq:generic-quantum-value} give
\begin{equation}
\xi_+=\frac{2+\sqrt2}{5}<\frac45=Q_\Phi.
\label{eq:qubit-steering-violation}
\end{equation}
Thus, both low-dimensional constructions also yield steering witnesses.
}

\section{Conclusions}
\label{sec:conclusions}
{
We have constructed an explicit family of informationally overcomplete measurements and introduced associated steering witnesses in every prime dimension. The main constructive result is an analytical interval ensuring semi-positivity of the $(1,2d-1)$-POVM elements in any prime dimension. For $d\geq5$, this interval follows from combining complementary spectral bounds from two unitarily equivalent representations of the noncentral operators, together with the exact spectrum of the central operator. The qubit and qutrit cases are treated separately but lead to the same conclusions.

Our construction provides an explicit positive measurement realization for the general GEAM steering framework from Ref.~\cite{RutkowskiSiudzinska2025}. A maximally entangled state and transposed, block-normalized effects yield an assemblage whose attained value $Q_\Phi$ exceeds the analytic LHS bound for the specified parameter choices. For $d\geq3$, the sufficient violation criterion is $Z(d,t,\tau)>1$. A qubit realization completes the result for all prime dimensions.
}

Further work may determine the exact positivity interval and the optimal quantum value $\beta_{\mathrm Q}(\mathcal F)$ for this family. Another direction is to optimize the measurement parameters and block weights for robustness of the resulting steering witnesses to noise.

\section{Acknowledgements}

This research was funded in whole or in part by the National Science Centre
Grant No.~2025/58/E/ST2/00336. For the purpose of Open Access, the author
has applied a CC-BY public copyright licence to any Author Accepted
Manuscript (AAM) version arising from this submission.


\onecolumngrid
\appendix

\section{Mutually unbiased bases and measurements}
\label{app:MUB-MUM}

For $d=2$, there exist a set of three mutually unbiased bases (MUBs),
\begin{equation}
\begin{split}
\mathcal{B}_0&=\{|0\>,\,|1\>\},\\
\mathcal{B}_1&=\left\{\frac{|0\>+|1\>}{\sqrt{2}},\,\frac{|0\>-|1\>}{\sqrt{2}}\right\},\\
\mathcal{B}_2&=\left\{\frac{|0\>+i|1\>}{\sqrt{2}},\,\frac{|0\>-i|1\>}{\sqrt{2}}\right\}.
\end{split}
\end{equation}
In the case of higher dimensions, the approach is different. For prime $d\geq 3$, $d+1$ sets of mutually unbiased bases (MUBs) $\mathcal{B}_\alpha=\{|\psi_{\alpha,k}\>\}$, $\alpha=0,\ldots,d$, are constructed as follows \cite{Ivanovic,Wiesniak},
\begin{align}
\mathcal{B}_0&=\left\{|k\>;\,k=0,\ldots,d-1\right\},\\
\mathcal{B}_\alpha&=\left\{\frac{1}{\sqrt{d}}\sum_{\ell=0}^{d-1}\omega^{\alpha\ell^2+k\ell}
|\ell\>;\,k=0,\ldots,d-1\right\},
\end{align}
where $\omega=\exp(2\pi i/d)$. Each basis defines a projective POVM
\begin{equation}
\mathcal{P}_\alpha=\{P_{\alpha,k}=U_\alpha P_kU_\alpha^\dagger;\,k=0,\ldots,d-1\},\qquad
\alpha=0,\ldots,d,
\end{equation}
where $P_k=|k\>\<k|$ are rank-1 projectors onto the operational basis vectors. The unitary operators $U_\alpha$ for $d=2$ are represented in the operational basis by the matrices
\begin{equation}\label{Ua2}
U_0=\mathbb{I}_2,\qquad 
U_1=\frac{1}{\sqrt{2}}\begin{pmatrix} 1 & 1 \\ 1 & -1 \end{pmatrix},\qquad
U_2=\frac{1}{\sqrt{2}}\begin{pmatrix} 1 & i \\ i & 1 \end{pmatrix},
\end{equation}
and for $d\geq 3$ they read
\begin{equation}\label{Ua}
U_0=\mathbb{I}_d,\qquad U_\alpha=\frac{1}{\sqrt{d}}\sum_{m,n=0}^{d-1}\omega^{\alpha m^2+mn} |m\>\<n|,
\end{equation}
As shown in ref. \cite{Kalev2014}, to each maximal set of MUBs corresponds a Hermitian operator basis $\mathcal{G}=\{G_0=\mathbb{I}_d/d,G_{\alpha,k}=U_\alpha G_kU_\alpha^\dagger;\,\alpha=0,\ldots,d;\,k=1,\ldots,d-1\}$ with traceless
\begin{equation}\label{Gk}
G_k=\frac{1}{\sqrt{d}(\sqrt{d}+1)}(\mathbb{I}_d+\sqrt{d}P_0)-P_k.
\end{equation}
For $d=2$, $\mathcal{G}$ are the normalized Pauli matrices.

To go beyond projective measurements, one introduces mutually unbiased measurements (MUMs) \cite{Kalev2014}. From definition, these are $d+1$ sets $\mathcal{E}_\alpha=\{E_{\alpha,k};\,k=0,\ldots,d-1\}$, $\alpha=0,\ldots,d$, of $d$ measurement operators $E_{\alpha,k}$ such that $\Tr (E_{\alpha,k})=1$ and
\begin{equation}
\begin{split}
\Tr (E_{\alpha,k}E_{\alpha,\ell})&=\kappa\delta_{k\ell}
+\frac{1-\kappa}{d-1}(1-\delta_{k\ell}),\\
\Tr (E_{\alpha,k}E_{\beta,\ell})&=\frac{1}{d},\qquad \beta\neq\alpha,
\end{split}
\end{equation}
{where $1/d<\kappa\leq 1$.}
Using the Hermitian operator basis $\mathcal{G}$, we construct the associated MUMs \cite{Kalev2014}:
\begin{equation}
E_{\alpha,0}=\frac{1}{d}\mathbb{I}_d+t_\alpha(\sqrt{d}+1)G_\alpha
\end{equation}
and
\begin{equation}
E_{\alpha,k}=\frac{1}{d}\mathbb{I}_d+t_\alpha\Big[G_\alpha-\sqrt{d}(\sqrt{d}+1)G_{\alpha,k}\Big],
\end{equation}
with $k=1,\ldots,d-1$, $G_\alpha=\sum_{\ell=1}^{d-1}G_{\alpha,\ell}$. The correspondence between $t_\alpha$ and $\kappa$ reads
\begin{equation}
t_\alpha^2=\frac{d\kappa-1}{d(d-1)(\sqrt d+1)^2},\qquad \kappa=\frac1d+(d-1)(\sqrt d+1)^2t_\alpha^2.
\end{equation}
{From the admissible range of $\kappa$, $t_\alpha$ is always non-zero.} Choosing the same $t_\alpha=t/[\sqrt d(\sqrt d+1)]$ for every $\alpha$, we obtain
\begin{equation}
E_{\alpha,k}=\frac{1-t}{d}\mathbb{I}_d+tP_{\alpha,k},\qquad -\frac{1}{d-1}\leq t\leq 1,
\end{equation}
where the MUBs are recovered for $t=1$. In the main text, we consider an incomplete set of $d-1$ MUMs $\mathcal{E}_\alpha=\{E_{\alpha,k};\,k=0,\ldots,d-1\}$ with $\alpha=1,\ldots,d-1$.

\section{Gershgorin bounds and admissible values of $\tau$}
\label{app:Gershgorin}
{
In the following calculations, we use the parameters $\mu,\nu,\xi$ from
Eq.~\eqref{eq:auxiliary-parameters} and
$\eta=\sqrt{d-2}$, $\chi=\sqrt{d-1}$, $\zeta=\xi/(d-1)$
from Eq.~\eqref{etas}.} All comparisons below hold for $d\geq5$.
The row endpoints follow directly from the two matrices in
Proposition~\ref{prop:spectral-reduction}.

\subsection{Lower endpoints}
For $B_k^{(1)}$, the lower endpoints are
\begin{align}
A_1&=\frac{\mu+\nu-\eta}{\sqrt d},\label{eq:appC-A1}\\
A_2&=\frac{\mu+\xi-\eta}{\sqrt d},\label{eq:appC-A2}\\
A_3&=\frac{\mu+d-2-2\eta}{\sqrt d}.\label{eq:appC-A3}
\end{align}
The {parameter} definitions imply
\begin{equation}
\frac{\xi}{\sqrt d}>2d-1>d+\sqrt d,
\qquad \frac{\xi}{\sqrt d+1}>d,
\qquad \nu<0.
\label{eq:appC-basic-xi}
\end{equation}
Consequently, $A_2>A_1$, and
\begin{equation}
A_3-A_1=\frac{\xi/(\sqrt d+1)-2-\eta}{\sqrt d}>0,
\end{equation}
since $d\geq2+\sqrt{d-2}$ for $d\geq3$. Thus,
\begin{equation}
\lambda_{\min}(B_k^{(1)})\geq A_1
=-\frac{d+1+\sqrt{d-2}+\sqrt{2d-1}}{\sqrt d}
=m_d^{(1)}.
\label{eq:appC-md1-explicit}
\end{equation}
For $B_k^{(2)}$, the lower endpoints are
\begin{align}
B_1&=\frac{\mu+\nu+1-\chi}{\sqrt d},\label{eq:appC-B1}\\
B_2&=\frac{\mu+\chi^2+\zeta-\chi-\eta\zeta}{\sqrt d},\label{eq:appC-B2}\\
B_3&=\frac{\mu+(\eta^2-\eta)\zeta}{\sqrt d}.\label{eq:appC-B3}
\end{align}
Their differences satisfy
\begin{align}
B_2-B_1&=\frac{2}{\sqrt d}
\left[\frac{\xi}{(d-1)(\sqrt d+\eta)}-1\right]>0,\\
B_3-B_1&=\frac1{\sqrt d}\left[
\xi\left(\frac{d-2-\eta}{d-1}+\frac1{\sqrt d+1}\right)
-d-1+\chi\right]>0.
\end{align}
The first inequality follows from
$\xi>(2d-1)\sqrt d>(d-1)(\sqrt d+\eta)$.
For the second, {note that} $d-2-\eta\geq0$,
$\xi/(\sqrt d+1)>d$, and $\chi>1$. Therefore,
\begin{equation}
\lambda_{\min}(B_k^{(2)})\geq B_1
=-\frac{d+\sqrt{d-1}+\sqrt{2d-1}}{\sqrt d}
=m_d^{(2)}.
\label{eq:appC-md2-explicit}
\end{equation}
Moreover, {one has}
\begin{equation}
m_d^{(2)}-m_d^{(1)}
=\frac{1+\sqrt{d-2}-\sqrt{d-1}}{\sqrt d}>0,
\label{eq:appC-minimum-comparison}
\end{equation}
{and hence the second basis gives the stronger lower bound.}

\subsection{Upper endpoints}
For $B_k^{(1)}$, the upper endpoints are
\begin{align}
U_1&=\frac{\mu+\nu+2+\eta}{\sqrt d},\label{eq:appC-U1}\\
U_2&=\frac{\mu+\xi+2+\eta}{\sqrt d},\label{eq:appC-U2}\\
U_3&=\frac{\mu+d-2+2\eta}{\sqrt d}.\label{eq:appC-U3}
\end{align}
We have $U_2-U_1=(\xi-\nu)/\sqrt d>0$ and
\begin{equation}
U_2-U_3=\frac{\xi-d+4-\eta}{\sqrt d}>0,
\end{equation}
where Eq.~\eqref{eq:appC-basic-xi} suffices for the last inequality.
Hence,
\begin{equation}
\lambda_{\max}(B_k^{(1)})\leq U_2
=\frac{\mu+\xi+2+\sqrt{d-2}}{\sqrt d}=M_d^{(1)}.
\label{eq:appC-Md1}
\end{equation}
For $B_k^{(2)}$, the upper endpoints are
\begin{align}
V_1&=\frac{\mu+\nu+1+\chi}{\sqrt d},\label{eq:appC-V1}\\
V_2&=\frac{\mu+\chi^2+\zeta+\chi+\eta\zeta}{\sqrt d},\label{eq:appC-V2}\\
V_3&=\frac{\mu+(\eta^2+\eta)\zeta}{\sqrt d}.\label{eq:appC-V3}
\end{align}
Direct subtraction gives
\begin{align}
V_2-V_1&=\frac1{\sqrt d}\left[
-2+\xi\left(1+\frac{\eta}{d-1}+\frac1{\sqrt d+1}\right)\right]>0,\\
V_2-V_3&=\frac1{\sqrt d}\left[d-1+\chi-\xi\frac{d-3}{d-1}\right]<0.
\end{align}
Indeed, $\xi>2\sqrt d(d-1)$ and, for $d\geq5$,
\begin{equation}
2(d-3)>\sqrt d+1>\frac{d-1+\sqrt{d-1}}{\sqrt d}.
\end{equation}
Thus,
\begin{equation}
\lambda_{\max}(B_k^{(2)})\leq V_3
=\frac{\mu+(d-2+\sqrt{d-2})\zeta}{\sqrt d}=M_d^{(2)}.
\label{eq:appC-Md2}
\end{equation}
To compare the two upper bounds, we write
\begin{equation}
\sqrt d\,[M_d^{(2)}-M_d^{(1)}]
=\frac{\xi(\eta-1)}{d-1}-(\eta+2).
\label{eq:appC-M2-M1}
\end{equation}
For $d\geq6$, $\eta\geq2$, $\sqrt d>2$, and
$\xi>2\sqrt d(d-1)$ imply that the right-hand side is larger than
$4(\eta-1)-(\eta+2)=3\eta-6\geq0$.
For $d=5$, it equals
$3\sqrt5(\sqrt3-1)-(\sqrt3+2)$, which is larger than
$5\sqrt3-8>0$. {Finally,}
\begin{equation}
M_d^{(1)}<M_d^{(2)}\qquad(d\geq5).
\label{eq:appC-upper-order}
\end{equation}

\subsection{Uniform bounds and positivity}
The scalar block also lies above $m_d^{(2)}$. Indeed,
\begin{equation}
\frac{\xi}{\sqrt d(\sqrt d+1)}<2\sqrt{2d-1},
\end{equation}
because $1+\sqrt{2d-1}<2(\sqrt d+1)$. Consequently,
\begin{equation}
\frac{\mu}{\sqrt d}-m_d^{(2)}
>\frac{d-2+\sqrt{d-1}-\sqrt{2d-1}}{\sqrt d}>0.
\end{equation}
The last inequality follows from
$d-2+\sqrt{d-1}>d-1\geq\sqrt{2d-1}$ for $d\geq5$.
The same expression gives
\begin{equation}
\lambda_{\min}(H_d)-m_d^{(2)}
=\frac{d-2+\sqrt{d-1}-\sqrt{2d-1}}{\sqrt d}>0.
\label{eq:appC-Hd-min-comparison}
\end{equation}
Clearly, $M_d^{(1)}>\mu/\sqrt d$. It remains to bound
$\lambda_{\max}(H_d)$. Put $s=\sqrt d$ for this comparison. Then
\begin{align}
M_d^{(1)}-\lambda_{\max}(H_d)
={}&(1+\sqrt{2s^2-1})\left[
\sqrt{2s^2-1}\left(1-\frac1{s(s+1)}\right)
-s-1+\frac2s\right]+\frac{\eta}{s}.
\label{eq:appC-M1-Hd-difference}
\end{align}
The {square} bracket is positive if
\begin{equation}
\sqrt{2s^2-1}(s^2+s-1)>(s+1)(s^2+s-2).
\end{equation}
Both sides are positive for $s\geq\sqrt5$. The difference of their squares is
\begin{align}
&(2s^2-1)(s^2+s-1)^2-(s+1)^2(s^2+s-2)^2\nonumber\\
&\quad=s^6-5s^4+2s^3+10s^2-2s-5\nonumber\\
&\quad=s^4(s^2-5)+2s(s^2-1)+5(2s^2-1)>0.
\end{align}
Thus, $M_d^{(1)}>\lambda_{\max}(H_d)$.
Combining these comparisons with the isospectrality of the active blocks gives
\begin{equation}
m_d^{(2)}\leq\lambda_{\min}(H_j)
\leq\lambda_{\max}(H_j)\leq M_d^{(1)}
\label{eq:appC-final-uniform-bounds}
\end{equation}
for every $j=1,\ldots,2d-1$.
Finally, the eigenvalues of $E_j$ are
$1/(2d-1)+\tau\lambda$ with $\lambda\in\spec(H_j)$.
Since each $H_j$ is nonzero and traceless,
\begin{equation}
E_j\geq0\quad\Longleftrightarrow\quad
-\frac1{(2d-1)\lambda_{\max}(H_j)}\leq\tau
\leq-\frac1{(2d-1)\lambda_{\min}(H_j)},
\label{eq:appC-one-operator-interval}
\end{equation}
{which is consistent with the results of Ref. \cite{GEAM}.}
The uniform bounds therefore yield Eq.~\eqref{eq:tau-interval}.
Moreover, $\sum_jH_j=0$ by their definitions, so $\sum_jE_j=\mathbb I_d$.
This proves Theorem~\ref{thm:tau-bound}.

\twocolumngrid
\bibliographystyle{apsrev4-2}
\bibliography{biblio_steering}

\end{document}